\documentclass[twocolumn, superscriptaddress, pra]{revtex4-1}
 
 \usepackage{amssymb}
 \usepackage{mathdots}
 \usepackage{dsfont}
 \usepackage{epsfig,epstopdf}
 \usepackage{amsmath, amsthm}
\usepackage{graphicx}
\usepackage{dcolumn}
\usepackage{bm}
\usepackage{color}

\begin{document}

\title{Genuinely Multipartite Concurrence of $N$-qubit X-matrices}
\author{S. M. Hashemi Rafsanjani}
\email{hashemi@pas.rochester.edu}
\affiliation{Rochester Theory Center and the Department of Physics \& Astronomy, University of Rochester, Rochester, New York 14627}

\author{M. Huber}
\affiliation{University of Bristol, Department of Mathematics, Bristol, BS8 1TW, U.K.}

\author{C. J. Broadbent}
\author{ J. H. Eberly}
\affiliation{Rochester Theory Center and the Department of Physics \& Astronomy, University of Rochester, Rochester, New York 14627}


\date{\today}

\begin{abstract}
We find an algebraic formula for the $N$-partite concurrence of $N$ qubits in an X-matrix. X-matrices are density matrices whose only non-zero elements are diagonal or anti-diagonal when written in an orthonormal basis. We use our formula to study the dynamics of the $N$-partite entanglement of $N$ remote qubits in generalized $N$-party Greenberger-Horne-Zeilinger (GHZ) states. We study the case in which each qubit interacts with a local amplitude damping channel. It is shown that only one type of GHZ state loses its entanglement in finite time; for the rest, $N$-partite entanglement dies out asymptotically.~Algebraic formulas for the entanglement dynamics are given in both cases. We directly confirm that the half-life of the entanglement is proportional to the inverse of $N$. When entanglement vanishes in finite time, the time at which entanglement vanishes can decrease or increase with $N$ depending on the initial state. In the macroscopic limit, this time is independent of the initial entanglement. 
\end{abstract}

\pacs{..........}

\maketitle

 \newcommand{\beq}{\begin{equation}}
 \newcommand{\eeq}{\end{equation}}
 \newcommand{\dg}[1]{#1^{\dagger}}
 \newcommand{\reci}[1]{\frac{1}{#1}}
 \newcommand{\ket}[1]{|#1\rangle}
 \newcommand{\nim}{\frac{1}{2}}
 \newcommand{\om}{\omega}
 \newcommand{\te}{\theta}
 \newcommand{\la}{\lambda}
 \newcommand{\beqa}{\begin{eqnarray}}             
 \newcommand{\eeqa}{\end{eqnarray}}               
 \newcommand{\nn}{\nonumber}                      
 \newcommand{\bra}[1]{\langle#1\vert}                 
 \newcommand{\ipr}[2]{\left\langle#1|#2\right\rangle}
  \newcommand{\up}{\uparrow}
   \newcommand{\down}{\downarrow}
     \newcommand{\dn}{\downarrow}         
\newtheorem{lemma}{Lemma}

\section{Introduction}
Even though entanglement was already promoted by \citet{schrodinger1935} as a fundamental aspect of quantum theory and in mathematics it predated quantum mechanics by decades \cite{schmidt}, its value as a resource for a wide range of potential applications was not appreciated until recently \cite{horodecki2009, nielsen:2000}. Furthermore, although its importance is now largely recognized, witnessing and quantifying the entanglement of arbitrary mixed states are still open problems. In only few-party cases do prescriptions exist for determining the entanglement of a mixed state \cite{PhysRevLett.77.1413,*Horodecki19961,wootters, PhysRevLett.85.2625, *PhysRevA.67.012307,*PhysRevLett.97.260502,*PhysRevA.78.012327}.

The problem becomes much more difficult for genuinely multipartite entanglement, entanglement shared between more than two parties. Multipartite entanglement ($N>3$) is thought to play an essential role in many phenomena including quantum metrology \cite{Giovannetti19112004} and quantum phase transitions \cite{PhysRevA.73.010305,*RevModPhys.80.517}. Furthermore, it is of fundamental importance to understand the dynamics of multipartite entanglement when the number of parties sharing entanglement approaches the macroscopic limit, that is, $N\rightarrow \infty$. Previous  studies of the dynamics of multipartite entanglement have utilized measures that fail to capture exactly when multipartite entanglement disappears \cite{PhysRevLett.92.180403,PhysRevLett.93.230501,PhysRevA.65.052327,PhysRevLett.100.080501}. 

An essential step was taken by \citet{PhysRevLett.100.080501}, who utilized the entanglement of different bipartitions of an $N$-qubit system to qualitatively study the scaling laws for the decay of multiqubit entanglement. The caveat to this approach lies in the fact that the entanglement of different bipartitions is a necessary but not sufficient condition for $N$-partite entanglement. The genuinely multipartite entanglement between $N$ parties can vanish before the entanglement of any of the bipartitions vanish. The lack of such analysis is mostly due to the fact that although there have been many attempts to solve the problem of determining the multipartite entanglement of a given state \cite{1367-2630-12-5-053002,PhysRevLett.106.190502,PhysRevA.83.062337} (see also references in the paper by~\citet{PhysRevA.83.062325}), an algebraic and/or numerically efficient prescription has not yet emerged. An algebraic prescription would be especially desirable since it can potentially open the door for a wide range of analytical investigations of entanglement dynamics.

Recently, based on previous works by~\citet*{PhysRevA.67.052107} and \citet{springerlink:10.1007/s11128-007-0052-7}, a new measure of multipartite entanglement, called genuinely multipartite (GM) concurrence, has been proposed, and it has been shown that GM concurrence is an entanglement monotone \cite{PhysRevA.83.062325}. This measure reduces to Wootters's original concurrence \cite{wootters} for two qubits.~Additionally, an algebraic formula for a lower bound of the GM concurrence has been found by \citet{PhysRevA.83.062325}. The lower bound, when calculated for a two-qubit X-matrix, matches the value of Wootters's concurrence.

The X-matrix of \citet{yu-2007} is a density matrix of $N$ qubits, written in an orthonormal product basis, whose non-zero elements are only diagonal or anti-diagonal. The concurrence of a two-qubit X-matrix takes a very simple form \cite{yu-2007} and that is why these two-qubit states have been extensively used in studying the dynamics of entanglement between two qubits in many scenarios \cite{yu-2007,yonac1,PhysRevLett.93.140404, yonac2007,*Cui200744,*Zhang2007274,*PhysRevA.77.054301,*PhysRevA.77.012117,*0953-4075-42-2-025501}. 

%

In view of the fact that the GM concurrence lower bound, derived in Ref.~\cite{PhysRevA.83.062325}, matches the exact value of concurrence for a two-qubit X-matrix, one might wonder if the lower bound might also be exact for more than a two-qubit X-matrix. In this paper we prove that this conjecture is correct. The lower bound provided by \citet{PhysRevA.83.062325} is realized by X-matrices. We thus present an \textit{algebraic} formula for the GM concurrence of an $N$-qubit X-matrix. This is our principal result and it enables analytic formulation of dynamics of $N$-partite entanglement in different scenarios. For illustration, we utilize our formula to directly examine the decay of $N$-qubit entanglement exposed to local decoherence channels.

\section{Genuinely multipartite concurrence of X states}
$N$-partite entanglement is defined by its opposite, biseparability. A pure $N$-partite system $\ket{\psi}\in\mathcal{H}_{1}\otimes\mathcal{H}_{2}\cdots\otimes\mathcal{H}_{N}$ is biseparable if there is a bipartition of the $N$ parties $\mathcal{H}_{1}\otimes\mathcal{H}_{2}\cdots\otimes\mathcal{H}_{N}=\mathcal{H}_{A}\otimes\mathcal{H}_{B}$, where $\mathcal{H}_{A}=\mathcal{H}_{j_{1}}\otimes\mathcal{H}_{j_{2}}\otimes\cdots\otimes \mathcal{H}_{j_{k}}$, $\mathcal{H}_{B}=\mathcal{H}_{j_{k+1}}\otimes\mathcal{H}_{j_{k+2}}\otimes\cdots\otimes \mathcal{H}_{j_{N}}$ for which $\ket{\psi}=\ket{\psi_{A}}\otimes \ket{\psi_{B}}$, and $\ket{\psi_{A}}\in \mathcal{H}_{A}$ and $\ket{\psi_{B}}\in \mathcal{H}_{B}$ \cite{PhysRevLett.104.210501,horodecki2009}. In other words, a pure state is biseparable if it has at least one pure marginal (reduced density matrix). An $N$-partite state that cannot be written as an ensemble of biseparable states is an $N$-partite entangled state.

Before introducing GM concurrence, let us introduce the set of all bipartitions of $N$ parties. Each bipartition is a division of the set $\{1,2,\dots, N\}$ into two non-overlapping and non-empty subsets. The set of all such bipartitions is denoted by $J=\{J_{1}, J_{2}, \dots J_{2^{N-1}-1}\}$. For example for $N=3$, there are three bipartitions $J_{1}=\{1|2,3\},~J_{2}=\{2|1,3\},~\textup{and}~J_{3}=\{3|2,1\}$, so that $J=\{J_{1},J_{2},J_{3}\}$.

For a pure state $\ket{\psi}$, to each of the elements of $J$ we can associate two reduced density matrices $\hat{A}(\ket{\psi},J_{j})$, and $\hat{B}(\ket{\psi},J_{j}) $ by tracing out either of the subsystems associated with that bipartition. The biseparability of a pure state can be determined by whether for any of the elements of $J$, $\hat{A}(\ket{\psi},J_{j})$ and $\hat{B}(\ket{\psi},J_{j})$ are pure. If so, $\ket{\psi}$ is biseparable. Thus the purity of the $j$th bipartition, denoted by $\Pi_{j}(\ket{\psi})$, is a key parameter.

For a pure state GM concurrence is then defined \cite{PhysRevA.83.062325} as, 
\begin{align}\nn
C_{GM}(\ket{\psi}):&=\min_{j}\sqrt{2}\sqrt{1-\Pi_{j}(\ket{\psi})}.
\end{align}
Clearly, $C_{GM}(\ket{\psi})\ge0$ and it is equal to zero if and only if $\ket{\psi}$ is a biseparable state. For a bipartite system this definition reduces to the I-concurrence \cite{PhysRevA.64.042315}.

To determine whether a mixed state $\hat{\rho}$ is biseparable, one has to determine whether $\hat{\rho}$ can be written as a convex sum of pure biseparable states. Thus one has to check all the ways $\hat{\rho}$ can be written as a convex sum of pure states (all pure state decompositions). Let us distinguish different pure state decompositions of $\hat{\rho}$ by assigning a continuous superscript, $\alpha$, to label them: 
\begin{align}
\hat{\rho}=\sum_{i} p_{i}^{(\alpha)}  \ket{\psi_{i}^{\alpha}}\bra{\psi_{i}^{\alpha}}.
\end{align}
To determine whether a particular pure state decomposition is a sum of biseparable states, we can calculate the average pure state GM concurrence for that particular $\alpha$
\begin{align}\nn
C_{\alpha}(\hat{\rho})&=\sum_{i} p_{i}^{(\alpha)} C_{GM}( \ket{\psi_{i}^{\alpha}})\\
&=\sum_{i} p_{i}^{(\alpha)} \biggl\{\min_{j}\sqrt{2}\sqrt{1-\Pi_{j}(\ket{\psi_{i}^{\alpha}})}\biggr\}.
\end{align}

Now we are ready to extend the definition of GM concurrence to all mixed states:
\begin{align}
C_{GM}(\hat{\rho})=\min_{\alpha} ~C_{\alpha}(\hat{\rho}).
\end{align}
If $C_{GM}(\hat{\rho})=0$ this means that there is an $\alpha$ for which $C_{\alpha}(\hat{\rho})=0$. Then $\hat{\rho}$ can be written as a sum of pure biseparable states, so $\hat{\rho}$ is biseparable. If $C_{GM}(\hat{\rho})>0$, there is no $\alpha$ for which the $\ket{\psi_{i}^{\alpha}}$'s are all biseparable and thus $\hat{\rho}$ is an $N$-partite entangled state.
The GM concurrence of $N$ parties, as defined in Ref.~\cite{PhysRevA.83.062325}, is a monotone of genuinely multipartite entanglement; it distinguishes between biseparable and $N$-partite entangled states and is convex, invariant under local unitary transformations, and non-increasing under local operations and classical communications (LOCC) \cite{PhysRevA.83.062325}.
The operational meaning of GM concurrence in terms of mutual information is explicitly discussed in Ref. \cite{PhysRevA.86.022319}. 

If the orthonormal basis for the X-matrix is $\{\ket{0,0,\dots,0},\ket{0,0,\dots,1},\dots,\ket{1,1,\dots,1}\}$, then one can always write an X-matrix in the form given below
\begin{align}
\hat{X}=\left(  \begin{array}{cccccccc}
    a_{1} & ¥ & ¥ & ¥ & ¥ & ¥ & ¥ & z_{1} \\ 
    ¥ & a_{2} & ¥  & ¥ & ¥ & ¥ & z_{2} & ¥ \\ 
    ¥ & ¥ & \ddots & ¥ & ¥ & \iddots & ¥ & ¥ \\ 
    ¥ & ¥ & ¥ & a_{n} & z_{n} & ¥ & ¥ & ¥ \\ 
    ¥ & ¥ & ¥ & z_{n}^{*} & b_{n} & ¥ & ¥ & ¥ \\ 
    ¥ & ¥ & \iddots & ¥ & ¥ & \ddots & ¥ & ¥ \\ 
    ¥ & z_{2}^{*} & ¥ & ¥ & ¥ & ¥ & b_{2} & ¥ \\ 
    z_{1}^{*} & ¥ & ¥ & ¥ & ¥ & ¥ & ¥ & b_{1} \\ 
  \end{array}
\right),
\end{align}
where $n=2^{N-1}$, and we require $|z_{i}|\le \sqrt{a_{i}b_{i}}$ and $\sum_{i}(a_{i}+b_{i})=1$ to ensure that $\hat{X}$ is positive and normalized. One can see why density matrices in this class are called X-matrices. It can be shown that the GM concurrence of an $N$-qubit X-matrix is given by
\begin{align}
C_{GM}= 2\max\{0,|z_{i}|-w_{i}\}, ~ i=0,1,\dots,n
\end{align}
where $w_{i}=\sum_{j\neq i}^{n}\sqrt{a_{j} b_{j}}$. A detailed proof of this result is given in the appendix.

\section{Robustness of N-partite entanglement}
Restricted forms of X-matrices of more than two qubits have been used in some recent studies of the dynamics of multipartite entanglement.~The entanglement measures utilized in these studies yield qualitative information about the multipartite entanglement \cite{0953-4075-41-15-155501,PhysRevA.82.032326,*PhysRevA.79.032322,*PhysRevA.79.012318}. Direct study of the dynamics of genuinely multipartite entanglement has been an elusive problem mainly due to the lack of an analytical measure of genuinely multipartite entanglement that is simple to calculate. Only for Greenberger-Horne-Zeilinger (GHZ) states that undergo pure dephasing, has there been a successful attempt that uses a geometric measure to give the exact dynamics of $N$-partite entanglement \cite{PhysRevA.78.060301}. Our GM concurrence formula provides an opening to quantitatively examine the conjectures of such studies and many other scenarios whenever the initial density matrix is an X-matrix and the X nature of the density matrix is robust in the dynamics. 

In the following we use our formula to directly study the dynamics  of the multiqubit entanglement shared by  $N$ qubits when each qubit is subjected to a local amplitude damping channel. This can represent, for example, the spontaneous decay of $N$ two-level atoms, each in a separate zero-temperature Markovian reservoir. For a two-level atom in zero-temperature Markovian reservoir, the evolution of ground and excited states, $\ket{g,0}$ and $\ket{e,0}$, is given by
\begin{align} \nn
&U(t)\ket{e,0}=A_{t}\ket{e,0}+B_{t}\ket{g,\bold{1}},\\
&U(t)\ket{g,0}=\ket{g,0},
\end{align}
where $U(t)$ is the local propagator,  $A_{t}=\sqrt{1-P_{t}}$, and $B_{t}=\sqrt{P_{t}}$. Although $P_{t}$ has a time dependence $P_{t}=1-e^{-\gamma t}$ where $\gamma$ is the damping rate, we can also think of $P_{t}$ as the probability of the decay and write everything as a function of $P$ instead of an explicit dependence on time. Thus, in the following we drop the explicit time dependence of $P$. The state $\ket{\bold{1}}$ denotes an excited state of a local reservoir.

We study the dynamics of the multiqubit entanglement that is shared initially by $N$-partite GHZ states:
\begin{align}\nn 
\ket{\Phi_{N}^{(k)},\alpha}=\cos\alpha \ket{e^{\otimes N-k}g^{\otimes k}}+\sin\alpha\ket{g^{\otimes N-k}e^{\otimes k}}
\end{align}
This is a GHZ state where either $(N-k)$'s of the qubits are initially excited and the rest are in their ground state or $k$ qubits are initially excited and the rest are in their ground states. We first study the $k=0$ case. We present a detailed analysis only for the three qubit case but the generalization to $N$ qubits is straightforward.
By tracing out the reservoirs we find the density matrix of the three atoms,
\begin{align}
\hat{\rho}_{\Phi_{3}^{(0)}}(t)=  \left(  \begin{array}{cccccccc}
     a_{1} & ¥ & ¥ & ¥ & ¥ & ¥ & ¥ &  z_{1} \\ 
    ¥ &  a_{2} & ¥ & ¥ & ¥ & ¥ & ¥ & ¥ \\ 
    ¥ & ¥ &  a_{2} & ¥ & ¥ & ¥ & ¥ & ¥ \\ 
    ¥ & ¥ & ¥ & b_{2} & ¥ & ¥ & ¥ & ¥ \\ 
    ¥ & ¥ & ¥ & ¥ &  a_{2} & ¥ & ¥ & ¥ \\ 
    ¥ & ¥ & ¥ & ¥ & ¥ & b_{2} & ¥ & ¥ \\ 
   ¥ & ¥ & ¥ & ¥ & ¥ & ¥ & b_{2} & ¥ \\ 
     z_{1} & ¥ & ¥ & ¥ & ¥ & ¥ & ¥ &  b_{1} \\ 
  \end{array}\right),
  \end{align}
  where
  \begin{align} \nn
 & a_{1}=\cos^{2}\alpha |A_{t}|^{6}, &b_{1}=\sin^{2}\alpha +\cos^{2} \alpha |B_{t}|^{6},\\ \nn
 & a_{2}=\cos^{2}\alpha |A_{t}^{2}B_{t}|^{2},& b_{2}=\cos^{2}\alpha |A_{t} B_{t}^{2}|^{2},\\ \nn
 &z_{1}= \sin\alpha\cos\alpha A_{t}^{3}.
\end{align}
 For an initial  $\ket{\Phi_{N}^{(0)},\alpha}$ state the concurrence reads
\begin{align} \label{conformula}
&C_{N}^{(0)}=\max\{0,Q_{N}^{(0)}\}, \\ \nn
&Q_{N}^{(0)}=2(\cos^{2}\alpha)~(1-P)^{\frac{N}{2}}\left(|\tan\alpha|-(2^{N-1}-1)P^{\frac{N}{2}}\right).
\end{align}
In Fig.~\ref{figure4}, we plot $Q_{N}^{(0)}$ versus $P$ for $N=\{2,10,100\}$ qubits. It confirms that the bulk of the initial entanglement dies out faster (at smaller $P$'s) as the number of qubits increases. In the $Q_{N}^{(0)}$ formula, the non-negative factor, $2(\cos^{2}\alpha)~(1-P)^{\frac{N}{2}}$, determines the decay of entanglement for $N\gg 2$, and one can show that for the amplitude damping channel the half-life of the entanglement depends on $N$ as
\begin{align}
P_{\textup{half-life}}\approx \frac{2\log 2}{N},
\end{align}
which is the same as the half-life of the coherence elements in the density matrix. We observe from Fig.~\ref{figure4} that the half-life of the entanglement decreases as the number of the qubits increases. One might expect a similar dependence on $N$ for the time at which the entanglement disappears completely. However, this is not always the case. In order to show this we solve the equation $Q_{N}^{(0)}= 0$ for the critical value of $P$, beyond which the concurrence is zero:
\begin{align}
P\ge\big(\frac{|\tan\alpha|}{2^{N-1}-1}\big)^{\frac{2}{N}}=P_{c}.
\end{align}

If $P_{c}<1$, then the entanglement has a finite life \cite{PhysRevLett.93.140404}. Otherwise the entanglement dies out asymptotically.
\begin{figure}[t]
\includegraphics[width=\columnwidth]{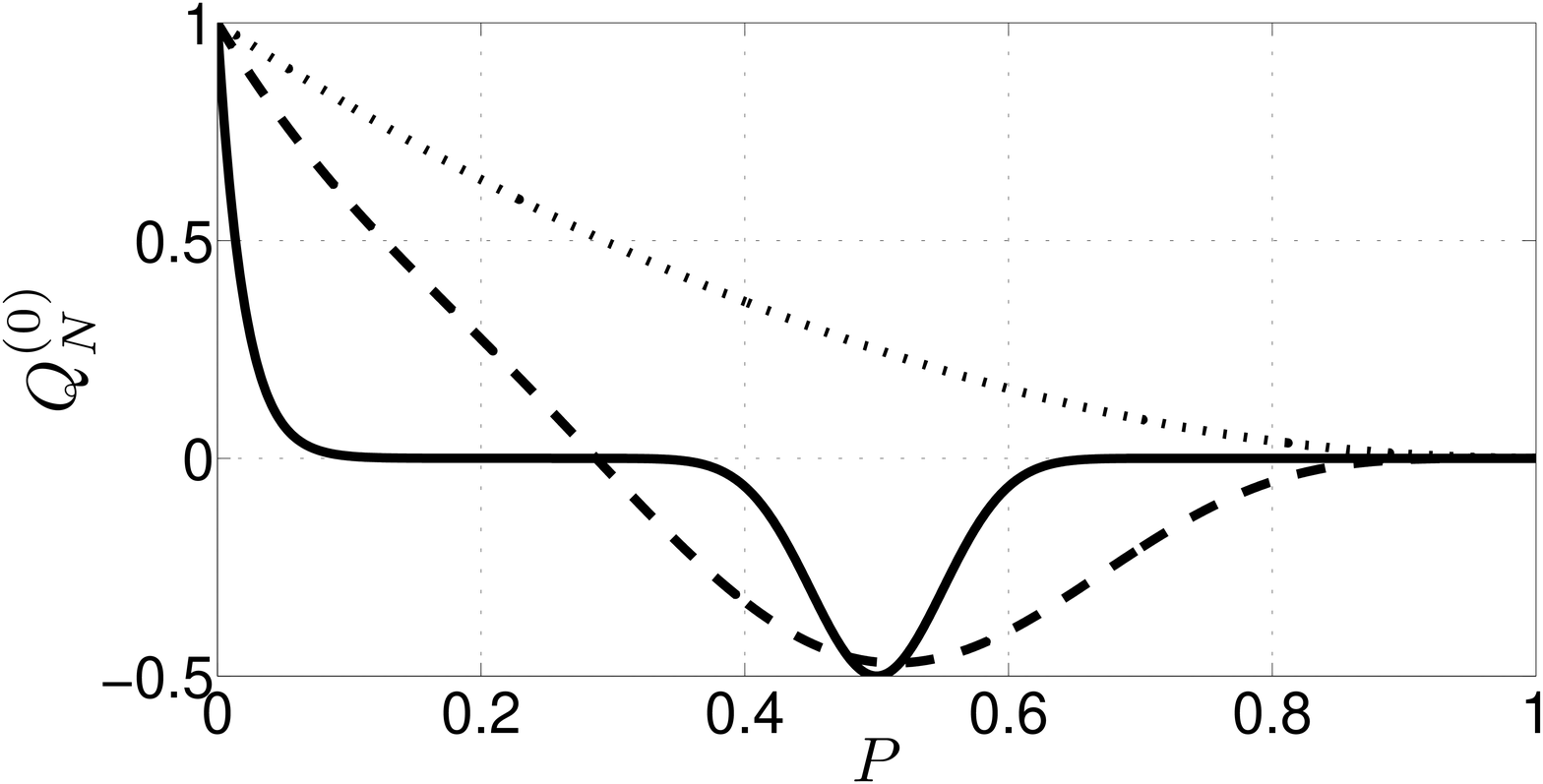}
\caption{The $Q_{N}^{(0)}$ vs $P$ for different numbers of qubits for the initial state $\ket{\Phi_{N}^{(0)},\frac{\pi}{4}}$. $N=2$ is the dotted line, $N=10$ is the dashed line and $N=100$ is the solid line.}
\label{figure4}
\end{figure}
In Fig.~\ref{figure5}, we plot $P_{c}$ versus the number of the qubits for different initial states. We observe that the critical value, $P_{c}$, can increase, decrease or even decrease and then increase as a function of $N$. The parameter that determines this peculiar dependence of $P_{c}$ on $N$ is $\tan\alpha$, which one can think of as a distance of the initial state to the final state. In the macroscopic limit, $N\rightarrow\infty$, even this dependence on $\tan\alpha$ is suppressed. Thus although the half-life of macroscopic entanglement  is very small, a non-zero entanglement lasts for a constant interval of time before vanishing completely.
\begin{figure}[b]
\includegraphics[width=\columnwidth]{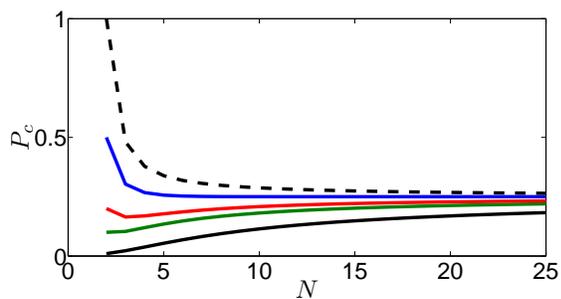}
\caption{$P_{c}$ vs number of qubits for different initial states $\ket{\Phi_{N}^{(0)},\alpha}$. From bottom up $\tan\alpha=0.01,0.1,0.2,0.5,$ and $1$ respectively.}
\label{figure5}
\end{figure}

Similar unusual behaviors were observed by \citet{PhysRevLett.100.080501} for the entanglement of different bipartitions of $N$ qubits. They had derived similar $N$-dependence for the half-life and also provided examples of initial states giving $P_{c}$ increasing with $N$. It should be pointed out that since the entanglement of different bipartitions is not a sufficient condition for N-partite entanglement similar behavior was not a foregone conclusion. 

For the $k>0$ case one can show that the initial entanglement only decays asymptotically. Below we present the argument for $N=3$ and $k=1$ but the extension to higher $N$ is straightforward. The asymptotic decay of entanglement is due to the fact that $\bra{e,e,e}\rho_{\Phi_{3}^{(1)}}\ket{e,e,e}$, $\bra{e,g,e}\rho_{\Phi_{3}^{(1)}}\ket{e,g,e}$, and $\bra{g,e,e}\rho_{\Phi_{3}^{(1)}}\ket{g,e,e}$ remain zero for all times. To show this, we note that our initial state is a superposition of two possibilities. Either two of the atoms, $\{1,2\}$, are excited and the other atom, $\{3\}$, is in vacuum state, or that single atom is excited and the atoms $\{1,2\}$ are in their ground states. Since all reservoirs are initially in their ground states, if an atom is in its ground state initially it will always remain there. The three diagonal terms that we referred to require the atom $\{3\}$ and at least one of the other two atoms to be excited simultaneously. Since this is forbidden, all of these matrix elements remain zero. Thus the negative contribution to the concurrence formula remain zero. This argument can be generalized for $N\ge 4$ to all of the GHZ states except for $\ket{\Phi_{N}^{(0)},\alpha}$, because in the $\ket{\Phi_{N}^{(0)},\alpha}$ state all of the atoms can be initially excited. The concurrence of all $\ket{\Phi_{N}^{(k)},\alpha}$ initial states with $k>0$ is given by
\begin{align}\nn
&C_{N}^{(k)}=|\sin2\alpha| (1-P)^{\frac{N}{2}}.
\end{align}
Thus for $k>0$, concurrence only dies when $P=1$. 
\section{conclusion}
N-partite entanglement, either as a resource for quantum computation or as a fundamental property of quantum theory, has been difficult to quantify, especially for mixed states. This is an important drawback since many of the algorithms  in quantum computation need N-partite entanglement between a large number of qubits, and inevitable interaction of these qubits with the environment renders initial pure states mixed and diminishes their entanglement. Thus it is of interest to understand, quantitatively, the dynamics of N-partite entanglement when the qubits sharing it come in contact with different environments. Here, we have found an algebraic formula for the genuinely multipartite (GM) concurrence of N-qubit density matrices that can be written as X-matrices in an orthonormal product basis. This development allows $N$-partite entanglement to be quantified for such states. The formula opens up the possibility of studying entanglement dynamics of $N$-qubit states in different scenarios, as long as the X-form of the density matrix is preserved. 

Using the concurrence formula, we have studied the dynamics of $N$-partite entanglement of $N$ two-level atoms interacting with local amplitude damping channels. We showed that only for $k=0$ the $\ket{\Phi_{N}^{(k)},\alpha}$ initial states lose their N-partite entanglement in finite time. Algebraic formulas for the concurrence were presented. It is observed that for large $N$ the bulk of initial concurrence decays with a rate inverse to $N$. For a given $N$ and $k=0$, the time at which entanglement vanishes to zero is determined by the distance of the initial state from the final state. In the macroscopic limit this time is independent of $\alpha $ too. An open question is whether this time interval also appears for other kinds of initial states in the macroscopic limit and whether it has any observable effect.

\section{acknowledgments}
S.M.H.R. acknowledges a useful communication with O. G\"uhne. We acknowledge partial financial support from ARO W911NF-09-1-0385 and NSF PHY-0855701. M.H. gratefully acknowledges support from the EC-project IP ``Q-Essence'', the ERC advanced grant ``IRQUAT'' and MC grant ``Quacocos''.

\section{Appendix: Concurrence of $N$-qubit X-states}
In Ref. \cite{PhysRevA.83.062325}, Ma \textit{et al.}~presented a lower bound for the GM concurrence. The lower bound of GM concurrence, derived in  Ref. \cite{PhysRevA.83.062325}, for an X-matrix is given by
\begin{align}
C_{GM}\ge 2\max\{0,|z_{i}|-w_{i}\}, ~ i=0,1,\dots,n
\label{lowerbound}
\end{align}
where $w_{i}=\sum_{j\neq i}^{n}\sqrt{a_{j} b_{j}}$. In the following we will show that this lower bound is exact for all X-matrices. Without loss of generality we can assume that $\sqrt{a_{1}b_{1}} \ge \sqrt{a_{i}b_{i}}$. Since we have assumed that $\sqrt{a_{1}b_{1}} \ge \sqrt{a_{i}b_{i}}$, it is easy to show that $|z_{i}|-w_{i}\le 0$ for $i>1$, so that Eq.~\ref{lowerbound} reduces to 
\begin{align}
C_{GM}(X)\ge 2\max\{0,|z_{1}|-w_{1}\}
\end{align}
We will show that this bound is actually an equality. First let us prove a lemma that we will utilize in our proof.

\begin{lemma}
\label{lemma}
The GM concurrence of an X-matrix for which $a_{1}b_{1}\ge a_{i}b_{i}$,~and $a_{j}=b_{j}=0$ for all $j\neq \{i,1\}$, is
\begin{align}
C_{GM}(\hat{X}_{1i})=2 \max\{0,|z_{1}|-\sqrt{a_{i}b_{i}}\}
\end{align}
\end{lemma}
\begin{proof}We already know that this quantity is a lower bound of GM concurrence. Thus we only need to show that it is also an upper bound. We do this by mapping $\hat{X_{1i}}$ to a two-qubit density matrix, $\hat{R}$, and then show that $C_{GM}(\hat{X}_{1i})$ is bounded from above by Wootters's concurrence of $\hat{R}$, where $C(\hat{R})=2\max\{0,|z_{1}|-\sqrt{a_{i}b_{i}}\}$.

Before going forward let us introduce some notation. Since we are working with qubits, we can represent each vector (ket) of the above basis as a number from 0 to $2^{N}-1$ written in the binary basis. For example, $\ket{0,0,\cdots,0}=\ket{0}$, $\ket{0,\cdots,0,1}=\ket{1}$, $\ket{0,\cdots,1,0}=\ket{2}$, and so on $\ket{1,\cdots,1,1}=\ket{2^{N}-1}$. We also denote the bit-flipped states in the same way, $\ket{\bar{i}}=\ket{2^{N}-i-1}$, for example, $\ket{\bar{0}}=\ket{2^{N}-1}$. In places where we need to label the individual qubits, we do so by using a subscript on the bits.

We perform the mapping by focusing on a specific bipartition of the qubits. The four non-zero diagonal elements of $\hat{X}_{1i}$ are $\{a_{1},a_{i},b_{1},b_{i}\}$, corresponding to projectors $\{\ket{0}\bra{0},\ket{i-1}\bra{i-1},\ket{\bar{0}}\bra{\bar{0}},~\textup{and}~\ket{\overline{i-1}}\bra{\overline{i-1}}\}$ respectively. Those qubits that contribute $1$ to the ket $\ket{i-1}$ we designate as party $F$. The rest of the qubits we denote as party $G$. For example, with seven qubits, which we denote as $(q_{1},q_{2},q_{3},q_{4},q_{5},q_{6},q_{7})$, where $i=6$, the basis states are
\begin{align} \nn
\ket{0}= &\ket{0_{1},0_{2},0_{3},0_{4},0_{5},0_{6},0_{7}},\\ \nn
\ket{5}= &\ket{0_{1},0_{2},0_{3},0_{4},1_{5},0_{6},1_{7}},\\ \nn
\ket{127}= &\ket{1_{1},1_{2},1_{3},1_{4},1_{5},1_{6},1_{7}},\\ 
\ket{122}= &\ket{1_{1},1_{2},1_{3},1_{4},0_{5},1_{6},0_{7}}.
\end{align}
Then party $F$ is given by qubits $(q_{1},q_{2},q_{3},q_{4},q_{6})$, and the remaining two qubits, $(q_{5},q_{7})$, make party $G$. Under this bipartition, we can write $\hat{X}_{1i}$ using the following basis states,
\begin{align}\nn
&\ket{\dn_{F}}=\ket{0_{1},0_{2},0_{3},0_{4},0_{6}},&~~&\ket{\up_{F}}=\ket{1_{1},1_{2},1_{3},1_{4},1_{6}},\\  \nn
&\ket{\dn_{G}}=\ket{0_{5},0_{7}},&~~~&\ket{\up_{G}}=\ket{1_{5},1_{7}},
\end{align}
\begin{align} \nn
\hat{X}_{1i}&=a_{1}\ket{\dn_{F}\dn_{G}}\bra{\dn_{F}\dn_{G}}+b_{1}\ket{\up_{F}\up_{G}}\bra{\up_{F}\up_{G}} \\ & + \nn
a_{i}\ket{\dn_{F}\up_{G}}\bra{\dn_{F}\up_{G}}+b_{i}\ket{\up_{F}\dn_{G}}\bra{\up_{F}\dn_{G}}\\ \nn
& +z_{1}\ket{\dn_{F}\dn_{G}}\bra{\up_{F}\up_{G}}+z_{1}^{*}\ket{\up_{F}\up_{G}}\bra{\dn_{F}\dn_{G}} \\ & +
z_{i}\ket{\dn_{F}\up_{G}}\bra{\up_{F}\dn_{G}}+z_{i}^{*}\ket{\up_{F}\dn_{G}}\bra{\dn_{F}\up_{G}},
\end{align}
We see that if we restrict attention to the subspace defined by the non-zero elements of $\hat{X}_{1i}$ we can map $\hat{X}_{1i}$ to a two qubit density matrix, $\hat{R}$, which, in the basis $\{\ket{\dn_{F}\dn_{G}},\ket{\dn_{F}\up_{G}},\ket{\up_{F}\dn_{G}},\ket{\up_{F}\up_{G}}\}$, reads
\begin{align}
\hat{X}_{1i}\longrightarrow \hat{R}=\left(  \begin{array}{cccc}
    a_{1} & ¥ & ¥ & z_{1} \\ 
    ¥ & a_{i} & z_{i} & ¥ \\ 
    ¥ & z_{i}^{*} & b_{i} & ¥ \\ 
    z_{1}^{*} & ¥ & ¥ & b_{1} \\ 
  \end{array}\right).
  \label{map}
\end{align}
Now that we have a two-qubit density matrix, we can take advantage of Wootters's concurrence. Note that from each pure-state decomposition (PSD) of $\hat{R}$ one can make a PSD of $\hat{X}_{1i}$ by mapping the basis states of the two-qubit system back to the multi-qubit basis states. We pick the PSD whose average concurrence is the minimum among all possible PSD's of  $\hat{R}$. Thus
\begin{align}
 \hat{R}=\sum_{i} p_{i} \ket{\psi_{i}}\bra{\psi_{i}},~~~C(\hat{R})=\sum_{i} p_{i} C(\ket{\psi_{i}}).
 \label{woottersconcurrence}
\end{align}
In Eq.~\ref{woottersconcurrence}, $C(\hat{R})$ is Wootters's concurrence, which, by definition, is equal to the minimum average concurrence over all possible PSD's of $\hat{R}$.
As mentioned before, each pure state $\ket{\psi_{i}}$ can be mapped back to an N-qubit state $(\ket{\psi_{i}} \rightarrow  \ket{\Psi_{i}} )$, producing a PSD for $\hat{X}_{1i}$, 
\begin{align}
 \hat{X}_{1i}=\sum_{i} p_{i} \ket{\Psi_{i}}\bra{\Psi_{i}}.
\end{align}
Since GM concurrence is convex by definition, we have
\begin{align}
 C_{GM}(\hat{X}_{1i})\le\sum_{i} p_{i} C_{GM}(\ket{\Psi_{i}}).
\end{align}
 For a pure state the GM concurrence is defined by 
 \begin{align}
C_{GM}(\ket{\Psi_{i}})=\min_{j} \sqrt{2}\sqrt{1-\Pi_{j}(\ket{\Psi_{i}})},
\end{align}
where the minimum is taken over all bipartitions, $J$, of the $N$ qubits. Therefore, the GM concurrence must be bounded by any specific bipartition, including the bipartition of the $N$ qubits to party $F$ and party $G$.
\begin{align} 
C_{GM}(\ket{\Psi_{i}})\le \sqrt{2}\sqrt{1- \Pi_{F|G} (\ket{\Psi_{i}})}.
\label{backmapping}
\end{align}
 Using the same mapping as for Eq.~\ref{map} it is easy to show that $C(\ket{\psi_{i}})$ is equal to the right hand side of Eq.~\ref{backmapping}. Therefore we conclude that 
\begin{align} \nn
 C_{GM}(\hat{X}_{1i})&\le\sum_{i} p_{i} C_{GM}(\ket{\Psi_{i}})\le \sum_{i} p_{i} C(\ket{\psi_{i}})\\
 &=C(\hat{R})=2\max\{0,|z_{1}|-\sqrt{a_{i}b_{i}}\},
\end{align}
where the right most equality is found by evaluating Wootters's concurrence for $\hat{R}$ under the assumption $a_{1}b_{1}\ge a_{i} b_{i}$. This upper bound matches the lower bound and therefore it is the exact value of $C_{GM}(\hat{X}_{1i})$.
\end{proof}

Next, we generalize this result to all the X-matrices. We do so by decomposing the X-matrix into a convex sum of $\hat{X}_{1i}$ matrices. Let us first look at the case for which $|z_{1}|- w_{1}\ge 0$.

(a) $|z_{1}|- w_{1}\ge 0$. Note that $|z_{i}|-w_{i}\le0$ for $i\ge 2$ since $\sqrt{a_{1}b_{1}}\ge\sqrt{a_{i}b_{i}}\ge|z_{i}|$. First by a change of phase of the basis, which is a local unitary transformation, we change $z_{1}$ to $|z_{1}|$. This only changes the phase of the other off-diagonal elements. Then we decompose $\hat{X}$ in the following form.
\begin{align}
&\hat{X}=\hat{A}+\sum_{i>1}^{n}\hat{S}_{i},
\end{align}
where $\hat{S}_{i}$ is an $\hat{X}_{1i}$ matrix whose two-qubit counterpart reads 
\begin{align}
\hat{R}_{i}=\left(  \begin{array}{cccc}
    x_{i} & ¥ & ¥ & \sqrt{a_{i}b_{i}} \\ 
    ¥ & a_{i} & z_{i} & ¥ \\ 
    ¥ & z_{i}^{*} & b_{i} & ¥ \\ 
    \sqrt{a_{i}b_{i}} & ¥ & ¥ & y_{i} \\ 
  \end{array}\right),
\end{align}
where
\begin{align}\nn
A_{11}&=a_{1}(1-\frac{w_{1}}{\sqrt{a_{1}b_{1}}}),~~~
A_{2n,2n}=b_{1}(1-\frac{w_{1}}{\sqrt{a_{1}b_{1}}}),\\ \nn
A_{1,2n}&=A_{2n,1}=|z_{1}|-w_{1},\\ \nn
A_{i,j}&=0~~i\neq \{1,2n\}, \text{or}~~ j\neq \{1,2n\},\\ 
x_{i}&=\frac{a_{1}\sqrt{a_{i}b_{i}}}{w_{1}},~~~~~~~~~
y_{i}=\frac{b_{1}\sqrt{a_{i}b_{i}}}{w_{1}}.
\end{align}
It can be shown that $\hat{S}_{i}$'s are all proportional to valid density matrices, since  they are non-negative hermitian matrices. The proportionality constant is between zero and one and can be interpreted as probability. Using Lemma \ref{lemma} one can show that all $\hat{S}_{i}$'s are biseparable matrices (though not normalized). Regarding the first matrix in the decomposition, the proportionality constant is $A_{11}+A_{2n,2n}$, and its GM concurrence is $2(|z_{1}|-w_{1})/(A_{11}+A_{2n,2n})$. Due to the convexity of GM concurrence we conclude that 
\begin{align}(A_{11}+A_{2n,2n})\frac{2(|z_{1}|-w_{1})}{A_{11}+A_{2n,2n}}=2(|z_{1}|-w_{1})\end{align}
 is an upper bound for the GM concurrence of $\hat{X}$. Since 2($|z_{1}|-w_{1}$) is also a lower bound for the concurrence, it is the exact value of the GM concurrence.
 
Note that for the above decomposition to work we had to assume that $|z_{1}|\ge w_{1}$. We now turn to the case $|z_{1}|< w_{1}$. We seek to show that all such density matrices are biseparable. We consider two different scenarios. 

(b) $\sqrt{a_{1} b_{1}}\ge w_{1}$. In this case the matrix $\hat{X}$ can be decomposed to matrices similar to the previous case. 
\begin{align}
\hat{X}=\sum_{i>1}^{n}\hat{S}_{i}^{\prime},
\end{align}
where $\hat{S}_{i}^{\prime}$ is an $\hat{X_{1i}}$ matrix whose two-qubit counterpart reads
\begin{align}
\hat{R}_{i}^{\prime}=\left(  \begin{array}{cccc}
    a_{1}T_{i} & ¥ & ¥ & z_{1}T_{i} \\ 
    ¥ & a_{i} & z_{i} & ¥ \\ 
    ¥ & z_{i}^{*} & b_{i} & ¥ \\ 
    z_{1}^{*}T_{i} & ¥ & ¥ &b_{1} T_{i} \\ 
  \end{array}\right), ~~T_{i}=\frac{\sqrt{a_{i}b_{i}}}{w_{1}}.
\end{align}
Since $|z_{1}|T_{i}\le \sqrt{a_{i}b_{i}}$~and~$|z_{i}|\le \sqrt{a_{1}b_{1}}T_{i},$ we can invoke Lemma 1 to confirm $\hat{R}_{i}^{\prime}$ is biseparable for all $i$. The fact that $\hat{R}_{i}^{\prime}$ is not normalized does not interfere with the proof of biseparability as one can always factor out $\textup{Tr}[\hat{R}_{i}^{\prime}]$. Now we focus on the last case.

(c) $\sqrt{a_{1} b_{1}}< w_{1}$. In this case we divide our matrix into two positive semi-definite matrices $\hat{X}=\hat{K}_{1}(t,r)+\hat{K}_{1}(r,t)$.
\begin{align}\nn
\hat{K_{1}}(t,r)=\left(  \begin{array}{cccccccc}
    a_{1}t & ¥ & ¥ & ¥ & ¥ & ¥ & ¥ & z_{1}t \\ 
    ¥ & a_{2}r & ¥  & ¥ & ¥ & ¥ & z_{2}r & ¥ \\ 
    ¥ & ¥ & \ddots & ¥ & ¥ & \iddots & ¥ & ¥ \\ 
    ¥ & ¥ & ¥ & a_{n}r & z_{n}r & ¥ & ¥ & ¥ \\ 
    ¥ & ¥ & ¥ & z_{n}^{*}r & b_{n}r & ¥ & ¥ & ¥ \\ 
    ¥ & ¥ & \iddots & ¥ & ¥ & \ddots & ¥ & ¥ \\ 
    ¥ & z_{2}^{*} r& ¥ & ¥ & ¥ & ¥ & b_{2}r & ¥ \\ 
    z_{1}^{*} t& ¥ & ¥ & ¥ & ¥ & ¥ & ¥ & b_{1}t \\ 
  \end{array}
\right),
\end{align}
where
\begin{align}\nn
 t=\frac{w_{1}}{w_{1}+\sqrt{a_{1}b_{1}}},~~\textup{and } ~~~r=1-t.
  \end{align}
  
   Note that since $w_{1}\le3\sqrt{a_{1}b_{1}}$, then $\frac{3}{4}\ge t>r$. One can show that 
\begin{align}
t\sqrt{a_{1}b_{1}}=r w_{1},
\end{align} 
which guarantees that $\hat{K}_{1}(t,r)$ falls in the category of case (a). Since 
\begin{align}\nn
&t|z_{1}|\le r w_{1},\\
&r|z_{j}|\le (t-r)\sqrt{a_{1}b_{1}}+r w_{j},
\end{align}
$\hat{K}_{1}(t,r)$ is biseparable. Regarding matrix $\hat{K}_{1}(r,t)$, since 
\begin{align}\nn
&r|z_{1}|\le t w_{1},\\
&t|z_{j}|\le r \sqrt{a_{1}b_{1}}+ t \sum_{i\neq 1,j}^{n} \sqrt{a_{i}b_{i}},
\end{align}
it does not  fall in the category of case (a) and thus belongs to either case (b) or case (c). If it falls in the category of case (b) then we can conclude that it is biseparable. If not, we divide $\hat{K}_{1}(t,r)$ into two matrices $\hat{K}_{1}(t,r)=\hat{K}_{2}(t^{\prime},r^{\prime})+\hat{K}_{2}(r^{\prime},t^{\prime})$, as before. Each time we divide a matrix in this way the trace of the remaining part is strictly smaller than the trace of the step before: $\textup{Tr}[\hat{K}_{i}(r^{i},t^{i})]\le 0.75^{i}$. Thus, we can write the matrix $\hat{X}$ as a convex sum of biseparable states and a remaining part that can be made arbitrarily close to zero. Therefore matrix $\hat{X}$ is a biseparable matrix. This completes the proof for all X-matrices. Therefore, we have proved that the GM concurrence of a $N$-qubit X-matrix is 
\begin{align}
C_{GM}(\hat{X})=2\max\{0,|z_{i}|-w_{i}\},~~1\le i\le n
\end{align}

\bibliographystyle{apsrev4-1}
\bibliography{mybib}

\end{document}